\documentclass[journal,draftcls,onecolumn,12pt,twoside]{IEEEtranTCOM}

\newcommand{\Wtilde}[1]{\stackrel{\sim}{\smash{#1}\rule{0pt}{1.1ex}}}

\usepackage{epsfig}
\usepackage{graphicx}
\usepackage{epstopdf}
\usepackage{amsmath}
\usepackage{amsfonts}
\usepackage{amssymb}
\usepackage{setspace}
\usepackage{pgf,pgfarrows,pgfnodes,pgfautomata,pgfheaps,pgfshade}
\usepackage{tikz}
\usepackage{amsthm}

\theoremstyle{plain}
\newtheorem{thm}{Theorem}[section]

\newtheorem{prop}[thm]{Proposition}

\newtheorem{defn}{Definition}[section]
\theoremstyle{remark}

\renewcommand{\thesection}{\arabic{section}}
\allowdisplaybreaks

\usepackage{amssymb}

\begin{document}
\title{Enhancing Secrecy Rates in a wiretap channel}

\author{Shahid M Shah,~\IEEEmembership{Student~Member,~IEEE,}
	and~Vinod~Sharma,~\IEEEmembership{Senior~Member,~IEEE}
\thanks{Part of the paper was presented in 2013 IEEE International Conference on Communications Workshop on Physical Layer Security (ICC), Budapest, Hungary.}
\thanks{Shahid M Shah and Vinod Sharma are with Electrical communication Department, Indian Institute of Science, Bangalore, India.}}
\maketitle




\begin{abstract}
Reliable communication imposes an upper limit on the achievable rate, namely the Shannon capacity. Wyner's wiretap coding, which ensures a security constraint also, in addition to reliability, results in decrease of the achievable rate. To mitigate the loss in the secrecy rate, we propose a coding scheme where we use sufficiently old messages as key and for this scheme prove that multiple messages are secure with respect to (w.r.t.) all the information possessed by the eavesdropper. We also show that we can achieve security in the strong sense. Next we consider a fading wiretap channel with full channel state information of the eavesdropper's channel and use our coding/decoding scheme to achieve secrecy capacity close to the Shannon capacity of the main channel (in the ergodic sense). Finally we also consider the case where the transmitter does not have the instantaneous information of the channel state of the eavesdropper, but only its distribution.
\end{abstract}

\begin{IEEEkeywords}
Physical layer security, Wiretap Channel, Resolvability, Rate loss

\end{IEEEkeywords}


\section{Introduction}\label{sec1}

With the advent of wireless communication, the issue of security has gained more importance due to the broadcasting nature of the wireless channel. Wyner \cite{wyner1975wire} proposed a coding scheme to implement security at physical layer for a degraded wiretap channel, which is independent of computational capacity of the adversary. The result of Wyner was generalized to a more general broadcast channel \cite{csiszar1978}. More recently, the growth of wireless communication systems has intensified the interest in implementing security at physical layer (\cite{liang2009information} , \cite{bloch2011physical}, \cite{liu2010securing}).

There is a trade-off between the achievable rate and the level of secrecy to be achieved. In particular, in the coding scheme which achieves the secrecy capacity in a discrete memoryless wiretap channel, the eavesdropper (Eve) is confused with the random messages at a rate close to Eve's channel capacity, thus resulting in loss of transmission rate \cite{wyner1975wire}, \cite{csiszar1978}.

Considerable progress has recently been made to improve the achievable secrecy rate of a wiretap Channel. In \cite{yamamoto1997rate} a wiretap channel with rate-distortion has been studied, wherein the transmitter and the receiver have access to some shared secret key before the communication starts. 
Secret key agreement between the transmitter (Alice) and the legitimate receiver (Bob) has been studied extensively in literature (\cite{ahlswede1993common}-\cite{prabhakaran2012secrecy}). When Alice and Bob have access to a public channel, the authors in \cite{ahlswede1993common} and \cite{maurer1993secret} proposed a scheme to agree on a secret key about which the adversary has less information (leakage rate goes to zero asymptotically).

In \cite{ardestanizadeh2009wiretap} the authors have considered the wiretap channel with secure rate limited feedback. This feedback is used to agree on a secret key, and the overall secrecy rate is enhanced. Under some conditions, the secrecy rate achieved can be equal to the main channel capacity. In \cite{lai2008wiretap} the authors have considered a modulo-additive discrete memoryless wiretap channel with feedback. The feedback is transmitted using the feed-forward channel only. The feedback signal can be used as a secret key. The authors propose a coding scheme which achieves secrecy rate equal to the main channel capacity. Wiretap channel with a shared key was studied in \cite{kang2010wiretap}.

Fading wiretap channel was studied in \cite{gopala2008secrecy}, \cite{liang2008secure} and \cite{bloch2008wireless}.
In \cite{khalil2013opportunistic}, previously transmitted confidential messages are stored in a secret key buffer and used in future slots to overcome the secrecy outage in a fading wiretap channel. In this model the data to be securely transmitted is delay sensitive. In \cite{gungor2013secrecy} the authors also use previously transmitted bits and store in a secret key buffer to leverage the secrecy capacity against deep fades in the main channel. The authors prove that all messages are secure w.r.t. all the outputs of the eavesdropper. The secrecy rate is not enhanced but prevented to decrease when the main channel is worse than the eavesdropper's channel.
A multiplex coding technique has been proposed in \cite{kobayashi2013secure} to enhance the secrecy capacity to the ordinary channel capacity. The mutual information rate between  Eve's received symbols and the (single) message transmitted is shown to decrease to zero as codeword length increases.

In most of the work cited above the security constraint used is \textit{weak secrecy} where if the message to be confidentially transmitted is $W$ and the information that eavesdropper gets in $n$ channel uses is $Z^n$, then $I(W;Z^n)\leq n\epsilon$. From stringent security point of view, this notion is proved to be vulnerable for leaking some useful information to the eavesdropper \cite{bloch2011physical}. Maurer in \cite{maurer1993secret} provided a coding scheme combined with privacy amplification and information reconciliation that achieves secrecy capacity (same as in the weak secrecy case) with a strong secrecy constraint, i.e., $I(W;Z^n)\leq \epsilon$. There are other ways to achieve strong secrecy (see chapter 21 in \cite{csiszar2011information}, \cite{devetak2005private} and \cite{bloch2011strong}).

In this paper, we consider a time slotted wiretap channel. The messages transmitted in a slot are used as a key to encrypt the message in the next slot of communication. Simultaneously, we use the wiretap encoder for another message in the same slot, which enhances the secrecy rate. We ensure that in each slot the currently transmitted message is secure with respect to (w.r.t.) all the output that Eve has received so far.

In next part of this paper we extend this work to the wiretap channel with a secret key buffer, where the key buffer is used to store the previously transmitted secret messages. In this scheme we use the oldest messages stored in the key buffer as a key in a slot and then remove those messages from the key buffer (a previous message is used as a key only once). In each slot this key is used along with a wiretap encoder to enhance the secrecy rate. With this, not only the current message but all the messages sent in recent past are jointly secure w.r.t. all the data received by Eve till now.
We also study a slow fading wiretap channel with the proposed coding scheme. We show that the water-filling power control along with our coding scheme provides the secrecy capacity close to Shannon capacity.

We also show that if \textit{resolvability} based coding scheme \cite{bloch2011strong} is used instead of \textit{wiretap} coding in a  slot, then we can achieve secrecy capacity equal to the main channel capacity in the strong sense also.

Rest of the paper is organised as follows. Channel model and the problem statement are presented in Section \ref{section_channel_model}. Section \ref{section_capacity} provides our coding and decoding scheme and shows that it can provide Shannon capacity for an AWGN wiretap channel. Section \ref{section_slow_fading} extends it to a fading wiretap channel with Eavesdropper's channel information at the transmitter while section \ref{section_no_csi} provides the results when this information is not available at the transmitter. Section \ref{section_conclution} concludes this paper.

A note about the notation. Capital letters, e.g., $W$ will denote a random variable and the corresponding small letter  $w$ its realization. An $n$-length vector $(A_1, A_2,\ldots, A_n)$ will be denoted as $\overline{A}$. Information theoretic notation will be same as in  \cite{el2011network}.
\section{Channel Model and Problem Statement}
\label{section_channel_model}
\begin{figure}[h]
	\setlength{\unitlength}{0.14in} 
	\centering 
	\begin{picture}(32,15) 
	\put(2,6.5){\framebox(5,2){Alice}}

	\put(9,4){\framebox(7,6.5){$P_{Y,Z\lvert X}(.\lvert .)$}}
	\put(18,8.5){\framebox(3,2){Bob}}
	\put(18,4.5){\framebox(3,2){Eve}}
	\put(0,7.5){\vector(1,0){2}}\put(7,7.5){\vector(1,0){2}}
	\put(16,9.5){\vector(1,0){2}}\put(16,5.5){\vector(1,0){2}}
	\put(21,9.5){\vector(1,0){2}}
	\put(21,5.5){\vector(1,0){2}}
	\put(0,7.8) {$\overline{W}_k$}

	\put(20,11) {$\widehat{\overline{W}}_k$}
	\put(20,3.5) {$R_L$}
	\put(7.3,7.8)
	{$\overline{X}_k$} \put(17,10.7) {$\overline{Y}_k$}
	\put(17,3.5){$\overline{Z}_k$}
		
	\put(2,12){\line(1,0){4}}
	\put(2,13){\line(1,0){4}}
	\put(3,12){\line(0,1){1}}
	\put(4,12){\line(0,1){1}}
	\put(5,12){\line(0,1){1}}
	\put(6,12){\line(0,1){1}}	
	\put(1,12.5){\vector(1,0){1.5}}
	\put(1,10.5){\line(0,1){2}}
	\put(6,12.5){\line(1,0){1.5}}
	\put(7.5,12.5){\vector(0,-1){2}}
	\put(4,11){$B_k$}
	\put(8,11){$\overline{R}_k$}
	\end{picture}
	\caption{Wiretap Channel with secret key buffer} 
	\label{fig:wiretap_buffer} 
\end{figure}
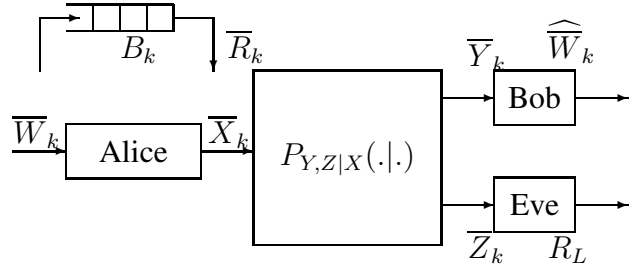

We consider a discrete time, memoryless, degraded wiretap channel, where Alice wants to transmit messages to Bob. We want to keep Eve (who is passively "listening") ignorant of the messages (Fig. 1).

Formally, Alice wants to communicate messages $W \in \mathcal{W}=\{1,2,\ldots,2^{nR_s}\}$ reliably over the wiretap channel to Bob, while ensuring that Eve is not able to decode them, where $R_s$, the secrecy capacity is defined below. $W$ is distributed uniformly over $\mathcal{W}$. At time $i$, $X_i$ is the channel input and Bob and Eve receive the channel outputs  $Y_i$ and $Z_i$ respectively, where $X_i \in \mathcal{X}, Y_i \in \mathcal{Y}, Z_i \in \mathcal{Z}$. The transition probability matrix of the channel is $p(y,z|x)$. The secrecy capacity (\cite{wyner1975wire})
\begin{equation}
R_s=\max_{p(x)}\left[I(X;Y)-I(X;Z)\right],
\end{equation}
is assumed $>0$. 

We consider the system as a time slotted system where each slot consists of $M+1$ minislots and one minislot consists of $n$ channel uses; $M$ being a large positive integer. We are interested in transmitting a sequence $\{W_m,m\geq 1\}$ of $iid$ messages uniformly distributed over $\mathcal{W}$.
Let $C$ be the capacity of Alice-Bob channel and $[x]$ denote the integer part of $x$. For simplicity, we take $\frac{C}{R_s}$ as an integer.
The message $\overline{W}_k$ to be transmitted in slot $k$ consists of one or more messages $W_m$. The codeword for message $\overline{W}_k$ is denoted by $\overline{X}_k$. The corresponding received bits by Eve are $\overline{Z}_k$. To increase the secrecy rate, the transmitter uses previous messages as keys for transmitting the messages in a later slot.
\par
We will denote by $P_e^{(n)}$ the probability that any of the messages transmitted in a slot is not received properly by Bob: $P_e^{(n)}=Pr(\overline{W}_k\neq \widehat{W}_k)$ where $\widehat{W}_k$ is the decoded message by Bob in slot $k$.
\par

For secrecy we consider the leakage rate 
\begin{align}
\frac{1}{n}I(\overline{W}_k,\overline{W}_{k-1},\ldots, \overline{W}_{k-N_1};\overline{Z}_1,\ldots,\overline{Z}_k)
\label{new_criteria}
\end{align}
 in slot $k$ where $N_1$ is an arbitrarily large positive integer which can be chosen as a design parameter to take into account the secrecy requirement of the application at hand \footnote[1]{One motivation for this is the law in various countries where old secret documents are declassified after a certain number of years.}. Then of course we should be considering $k>N_1$.This means that the Eve at time $k$ is not interested in \textit{very old} messages transmitted before slot $k-N_1$.
\par

\begin{defn}
Rate $R$ is achievable if there are coding-decoding schemes for each $n$ such that $P_e^{(n)}\rightarrow 0$ and $\frac{1}{n}I(\overline{W}_{k},\ldots,\overline{W}_{k-N_1}; \overline{Z}_1,\ldots,\overline{Z}_{k})\rightarrow 0$ as $n\rightarrow \infty$, where $N_1$ is an arbitrarily large fixed constant.
\end{defn}
In the following we explain our coding scheme. The message $\overline{W}_k$ transmitted in slot $k$ is stored in a key buffer (of infinite length) for later use as a key. After certain bits from the key buffer are used as a key for data transmission, those bits are discarded from the key buffer, not to be used again. Let $B_k$ be the number of bits in the key buffer at the beginning of slot $k$. Let $\overline{R}_k$ be the number of key bits used in slot $k$ from the key buffer. Then
\begin{equation}
B_{k+1}=B_k+\lvert \overline{W}_k\rvert-\overline{R}_k
\end{equation} 
where $\lvert \overline{W}_k\rvert$ denotes the number of bits in $\overline{W}_k$. Now we  explain the coding-decoding scheme used in this paper.

\subsection{Encoder:} To transmit message $\overline{W}_{k}$ in slot $k$, the encoder has two parts
\begin{equation}
f_s:\mathcal{W}\rightarrow \mathcal{X}^{n}, f_d:\mathcal{W}^M \times \mathcal{K} \rightarrow \mathcal{X}^{nM},
\end{equation}
where $\mathcal{K}$ is the set of secret keys generated and $f_s$ is the wiretap encoder, as in \cite{wyner1975wire}.
We use the following encoder for $f_d$: Take binary version of the message and $XOR$ with the binary version of the key. Encode the resulting encrypted message with an optimal usual channel encoder (e.g., an efficient LDPC code).

Assume $B_0=0$. The case of $B_0>0$ can be easily handled in the same way. In the first slot message $\overline{W}_1=W_1$, encoded using the wiretap coding only is transmitted (we use only the first minislot, see Fig. \ref{leakage_slotwise}) . At the end of slot 1, $nR_s$ bits of this message are stored in the key buffer. Thus $B_1=R_sn$. In slot 2, message $\overline{W}_2$ consisting of two messages $(\overline{W}_{21}, \overline{W}_{22})=(W_2,W_3)$ are transmitted. $W_2$ is transmitted via wiretap coding and $W_3$ uses $\overline{W}_1$ as a key and the encrypted message $\overline{W}_1\oplus W_3$ is transmitted via a usual capacity achieving channel code. At the end of slot 2, $R_sn$ bits of $\overline{W}_1$ are removed from the key buffer and $2R_sn$ bits of $\overline{W}_2$ are stored in the key buffer. Since Bob is able to decode $\overline{W}_1$ with a large probability, but not Eve, $\overline{W}_1$ can be an effective key in slot 2. In slot 3, message $\overline{W}_3$ consisting of 3 messages from the source message sequence are transmitted: one message in the first mini slot denoted as $W_{3,1}$ via wiretap coding and two messages denoted together as $\overline{W}_{3,2}$ via encryption with  message $\overline{W}_{2}$ as key bits. In any mini-slot we can transmit upto $C/R_s$ messages via encryption with a key. This is because we cannot transmit reliably at a rate higher than Bob's capacity $C$. Thus, the maximum number of messages that can be transmitted in a slot is $1+\frac{C}{R_s}M \triangleq M_1$. Once we reach this limit, from then onwards $ M_1$ messages will be transmitted in a slot providing the achievable rate $\frac{R_s+CM}{M+1}$ which can be made as close to $C$ as we wish by making $M$ arbitrarily large.

Consequently, in slot $k\leq  M_1$, $k$ messages from the source message stream are transmitted, $(k-1)R_sn$ bits from the key buffer are removed in the beginning of slot $k$ and $kR_sn$ bits are added to the key buffer at the end of slot $k$. The overall message is denoted by $\overline{W}_k=\left(\overline{W}_{k,1},\overline{W}_{k,2}\right)$ with $\overline{W}_{k,1}$ consisting of one source message transmitted via wiretap coding and $\overline{W}_{k,2}$ consisting of $k-1$ messages transmitted via the secret key. From slot $M_1$ onwards $M_1$ messages are transmitted in the above mentioned fashion.

We use the key buffer as a first in first out (FIFO) queue, i.e., at any time the oldest key bits in the buffer are used first. Also $B_k \rightarrow \infty$ as $k \rightarrow \infty$. 
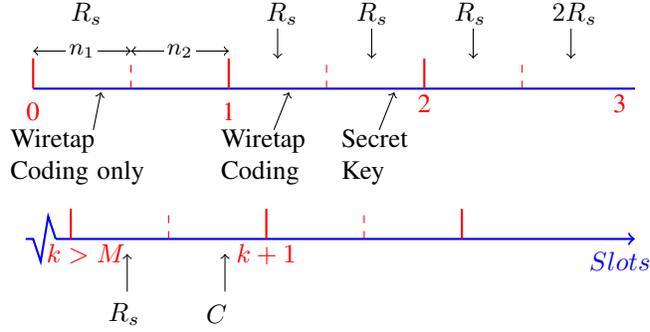
\begin{figure}
		\begin{tikzpicture}
	
	\tikzset{input/.style={}}
	\tikzset{block/.style={rectangle,draw}}
	\tikzstyle{pinstyle} = [pin edge={to-,thick,black}]
	
	\tikzstyle{ann} = [fill=white,font=\footnotesize,inner sep=1pt]
	\draw[blue,thick] (0,0) -- (8,0);
	\draw[blue,thick] (-0.1,-2.0)--(0,-2.0)--(0.1,-2.3)--(0.2,-1.7)--(0.3,-2.0);
	\draw[blue,thick, arrows=->](0.3,-2)--(8,-2);    
	\draw[red,thick] (0,0) -- (0,0.4);
	\draw[red,thin,dashed](1.3,0)--(1.3,0.4);
	\draw[red,thick] (2.6,0) -- (2.6,0.4);
	\draw[red,thin,dashed](3.9,0)--(3.9,0.4);
	\draw[red,thick] (5.2,0) -- (5.2,0.4);
	\draw[red,thin,dashed](6.5,0)--(6.5,0.4);
	\draw[red,thick] (0.5,-2.0) -- (0.5,-1.6);
	\draw[red,thin,dashed] (1.8,-2.0)--(1.8,-1.6);
	\draw[red,thick] (3.1,-2.0) -- (3.1,-1.6);         
	\draw[red,thin,dashed] (4.4,-2.0) -- (4.4,-1.6);         	
	\draw[red,thick](5.7,-2.0)--(5.7,-1.6);
	\node[text width=1.8cm, font=\small ] at (0.6,-0.9) {Wiretap Coding only};     
	\draw[arrows=->] (0.8,-0.55)--(0.9,-0.05);
	\draw[arrows=<->](0,0.5)--(1.3,0.5);   
	\node[ann] at (0.65,0.5) {$n_1$}; 
	\draw[arrows=<->](1.3,0.5)--(2.6,0.5);   
	\node[ann] at (1.95,0.5) {$n_2$}; 
	\node[text width=1.2cm, font=\small] at (3.1,-0.9) {Wiretap Coding}; 
	\node[text width=1.2cm, font=\small] at (4.7,-0.9) {Secret Key};   
	\draw[arrows=->] (3.2,-0.55)--(3.4,-0.05);
	\draw[arrows=->] (4.6,-0.55)--(4.8,-0.05);
	\node[red, font=\small] at (0,-0.3) {0};     \node[red,font=\small] at (2.6,-0.3) {1};  
	\node[red, font=\small] at (5.2,-0.2) {2};     \node[red,font=\small] at (7.8,-0.2) {3}; 
	\node[red, font=\small] at (0.7,-2.2) {$k>M$};     \node[red,font=\small] at (3.1,-2.2) {$k+1$};    
	\node[blue,font=\small] at (7.8,-2.3) {$Slots$};    
	\node[text width=1cm, font=\small] at (1,1){$R_s$};       
	\node[text width=1.2cm, font=\small ] at (3.7,1) {$R_s$};      
	\node[text width=1.2cm, font=\small ] at (4.9,1) {$R_s$};        
	\node[text width=1.2cm, font=\small ] at (6.2,1) {$R_s$};    
	\node[text width=1.2cm, font=\small ] at (7.5,1) {$2R_s$};  
	
	\node[text width=1.2cm, font=\small ] at (1.6,-3.0) {$R_s$};  
	\node[text width=1.2cm, font=\small ] at (2.9,-3.0) {$C$};          
	\draw[arrows=<-](3.25,0.4)--(3.25,0.8);   
	\draw[arrows=<-](5.85,0.4)--(5.85,0.8);   
	\draw[arrows=<-](4.5,0.4)--(4.5,0.8);       
	\draw[arrows=<-](7.15,0.4)--(7.15,0.8);
	\draw[arrows=->](1.25,-2.7)--(1.25,-2.2);          
	\draw[arrows=->](2.55,-2.7)--(2.55,-2.2);       
	\end{tikzpicture}
\caption{Coding Scheme to achieve Shannon Capacity in Wiretap Channel}
		\label{leakage_slotwise}
\end{figure}

\subsection*{Decoder} We have a secret key buffer at Bob's decoder also that is used in the same way as at the transmitter. The confidential messages decoded by the decoder are stored in this buffer. For decoding at Bob, in slot 1 the usual wiretap decoder is used (say, a joint-typicality decoder). From slot 2 onwards, for the first mini-slot, we use the wiretap decoder while for the rest of the mini-slots, we use the channel decoder (corresponding to the channel encoder used) and then $XOR$ the decoded message with the key used.

 The above coding-decoding schemes ensure that $P_e^{(n)}\rightarrow 0$ as $n\rightarrow\infty$.  There is a small issue of error propagation due to using the previous message as key: Let $\epsilon_n$ be the message error probability for the wiretap encoder and let $\delta_n$ be the message error probability due to the channel encoder for $\overline{W}_{k}$. Then $\epsilon_n \rightarrow 0$ and $\delta_n \rightarrow 0$ as $n\rightarrow \infty$. For the $k^{th}$ slot, we have $P(\overline{W}_k\neq\widehat{W}_k)\leq Pr($Error in decoding $ \overline{W}_{k1})+ Pr($Error in decoding $ \overline{W}_{k2})+Pr($Error in decoding $ \overline{W}_{k-1}) \leq k\epsilon_n + (k-1)\delta_n$. Thus the error increases with $k$. But restarting (as in slot 1) after some large $k$ slots  as in slot 1 (i.e., again start with one message in the first minislot and no message in the rest of the slot) will ensure that $P(\overline{W}_k\neq\widehat{W}_k)\rightarrow 0$ as $n\rightarrow \infty$.

In the rest of the paper we show that our coding scheme provides an achievable rate with the above secrecy criterion as close to $C$ as needed, for all $k$ large enough. We also note that the following proof is valid for $N_1>1$. The proof for $N_1=1$ is different from the proposed proof and one can refer to \cite{shah2013previous} for details of the proof. We will denote the codeword $\overline{X}_k=(\overline{X}_{k,1},\overline{X}_{k,2})$ and $\overline{Z}_k=(\overline{Z}_{k,1},\overline{Z}_{k,2})$ for the data received by Eve in the first part and the second part of slot $k$.

\section{Capacity of Wiretap Channel}
\label{section_capacity}

\begin{thm}
	The secrecy capacity of our coding-decoding scheme is $C$ and it satisfies (\ref{new_criteria}) for any $N_1\geq 0$, for all $k$ large enough.
\end{thm}

\textbf{\textit{Proof}}: As mentioned in the last section, by using our coding-decoding scheme, using wiretap coding and secure key, in any slot $k$, Bob is able to decode the message $\overline{W}_k$ with probability $P_e^{(n)}\rightarrow 0$ as $n \rightarrow \infty$.

Fix $N_1\geq 0$ and a small $\epsilon>0$. Due to wiretap coding, we can choose $n$ such that $I(\overline{W}_{k,1};\overline{Z}_{k,1})\leq n\epsilon$ for all $k\geq 1$. Since the key buffer $B_k\rightarrow \infty$, we use the oldest key bits in the buffer first and in any slot do not use more than $MC$ key bits, after sometime (say $N_2$ slots) for all $k\geq N_2$ we will be using key bits only from the messages $\overline{W}_1, \overline{W}_2,\ldots,\overline{W}_{k-N_1-1}$ for messages $\overline{W}_k,\overline{W}_{k-1},\ldots,\overline{W}_{k-N_1}$.  Furthermore,
\begin{align}
&I(\overline{W}_k, \overline{W}_{k-1},\ldots,\overline{W}_{k-N_1};\overline{Z}_1,\ldots, \overline{Z}_k) \nonumber \\
&=I(\overline{W}_{k,1}, \overline{W}_{k-1,1},\ldots,\overline{W}_{k-N_1,1};\overline{Z}_1,\ldots, \overline{Z}_k) \nonumber \\
&~~+I(\overline{W}_{k,2},\ldots,\overline{W}_{k-N_1,2};\overline{Z}_1,\ldots, \overline{Z}_k|\overline{W}_{k,1},\ldots,\overline{W}_{k-N_1,1}). \label{eqn_thm_1}
\end{align}  
We show in Lemma 1 that
\begin{equation}
I(\overline{W}_{k,1}, \overline{W}_{k-1,1},\ldots,\overline{W}_{k-N_1,1};\overline{Z}_1,\ldots, \overline{Z}_k)\leq (N_1+1)n\epsilon,  \label{eqn_thm_2}
\end{equation}
and in Lemma 2 that

\begin{equation}
I(\overline{W}_{k,2},\ldots,\overline{W}_{k-N_1,2};\overline{Z}_1,\ldots, \overline{Z}_k|\overline{W}_{k,1},\ldots,\overline{W}_{k-N_1,1})=N_1\epsilon.      \label{eqn_thm_3}
\end{equation}
From (\ref{eqn_thm_1}), (\ref{eqn_thm_2}) and (\ref{eqn_thm_3})

\begin{equation}
\frac{1}{n}I(\overline{W}_k, \overline{W}_{k-1},\ldots,\overline{W}_{k-N_1};\overline{Z}_1,\ldots, \overline{Z}_k)\leq (2N_1+1)\epsilon.
\end{equation}

By fixing $N_1$, we can take $\epsilon$ small enough such that $(N_1+1)\epsilon$ is less than any desired value.       $\mspace{230mu}$         $\square$

So far we have been considering an infinite buffer system. But an actual system will have a finite buffer. Now we compute the key buffer length needed for our system. 
\par
If we fix the probability of error for Bob and the upper bound on equivocation, then we can get the code length $n$ needed. Also, from the secrecy requirement, we can fix $N_1$. Once $n$ and $N_1$ are fixed, to ensure that eventually, in slot $k$ we will use a key from messages before time $k-N_1$, the key buffer size should be $\geq CMN_1n$ bits. Also, since in each slot, the key buffer length increases by $nR_s$ bits, the key buffer will have at least $CMN_1n$ bits after slot $\frac{CMN_1}{R_s}$. In the finite buffer case eventually key buffer will overflow. We should loose only the latest bits arriving in any slot (not the bits already stored).

We can obtain Shannon capacity even with strong secrecy. For this instead of using the usual wiretap coding of Wyner in the first minislot of each slot we use the \textit{resolvability} based coding scheme \cite{bloch2011strong}. Then $I(\overline{W}_{k,1};\overline{Z}_{k,1})\leq\epsilon$ instead of $I(\overline{W}_{k,1};\overline{Z}_{k,1})\leq n\epsilon$ for $n$ large enough. Then from proof of Theorem 1, our coding-decoding scheme provides
\begin{equation}
I(\overline{W}_k,\ldots,\overline{W}_{k-N_1};\overline{Z}_1,\ldots,\overline{Z}_k)\leq \epsilon.
\end{equation}

\section{AWGN Slow Fading Channel}
\label{section_slow_fading}
We consider a slow flat fading AWGN channel (Fig. 1), where the channel gains in a slot are constant. The channel outputs are,
\begin{equation}
Y_i=\Wtilde{H}X_i+N_{1i},
\end{equation}
\begin{equation}
Z_i=\Wtilde{G}X_i+N_{2i},
\end{equation}
where $X_i$ is the channel input, $\{N_{1i}\}$ and $\{N_{2i}\}$ are independent, identically  distributed $(i.i.d.)$ sequences independent of each other and $\{X_i\}$ with distributions $\mathcal{N}(0,\sigma_1^2)$ and $\mathcal{N}(0,\sigma_2^2)$ respectively, and $\mathcal{N}(a,b)$ denotes Gaussian distribution with mean $a$ and variance $b$. Also $\Wtilde{H}$ and $\Wtilde{G}$ are the channel gains to Bob and Eve respectively in the given slot. Let $H=\lvert \Wtilde{H}\rvert ^2$ and $G=\lvert \Wtilde{G}\rvert ^2$.

The channel gains $H_k$ and $G_k$ in slot $k$ are constant and sequences $\{H_k,k\geq 0\}$ and $\{G_k,k\geq 0\}$ are $iid$ and independent of each other. We assume that $(H_k,G_k)$ is known at the transmitter and Bob at the beginning of slot $k$. The notation and assumptions are same as in Section 3. Power $P(H_k,G_k)$ is used in slot $k$ for transmission. There is an average power constraint,
\begin{equation}
\limsup_{k\rightarrow \infty}\frac{1}{k}\sum_{m=1}^{k}\mathsf{E}\left[P(H_k,G_k)\right] \leq \overline{P}.  \label{pow_cons}
\end{equation}

Given $H_k, G_k,$ and $B_k$ at the beginning of slot $k$, Alice needs to decide on $P(H_k,G_k)$ and $\overline{R}_k$ such that the resulting average transmission rate $\limsup_{k\rightarrow \infty}\frac{1}{k}\sum_{l=1}^{k}r_l$ is maximized subject to (\ref{pow_cons}), (\ref{new_criteria}) and $P_e^n\rightarrow 0$, where $r_k$ is the transmission rate in slot $k$. We compute this capacity for $P(H_k>G_k)>0$; otherwise, the capacity is zero. At the end of slot $k$, $n(M+1)r_k \triangleq \overline{r}_k$ bits are stored in the key buffer for later use as a key while $\overline{R}_k$ bits have been removed. Thus, the buffer size evolves as,
\begin{equation}
B_{k+1}=B_k+\overline{r}_k-\overline{R}_k.    \label{que}
\end{equation}

For convenience, we define 

\begin{equation}
C(P(H,G))=\frac{1}{2}\log \left(1+\frac{HP(H,G)}{\sigma_1^2}\right),
\end{equation}
and
\begin{equation}
C_e(P(H,G))=\frac{1}{2}\log \left(1+\frac{GP(H,G)}{\sigma_2^2}\right),
\end{equation}
where $P(H,G)$ is the power used when the channel gains are $H$ and $G$. Unlike Sections II and III where initial messages $W$ are with cardinality $2^{nR_s}$, we use adaptive coding and power control.
Then, we have the following theorem.
\\
\begin{thm}
	The secrecy rate 
	
	\begin{equation}
	\label{theorem_1}
	C_s=\mathsf{E}_{H}\left[C(P(H))\right]
	\end{equation}
	
	is achievable if $Pr(H_k>G_k)>0$, where $P(H)=P(H,G)$ is the water-filling power policy for Alice $\rightarrow$ Bob channel.
\end{thm}

\begin{proof}
	We follow the coding-decoding scheme of Section-\ref{section_capacity} with the following change to account of the fading.
	
	Each slot has $M+1$ mini-slots. We fix a power control policy $P(H,G)$ satisfying average power constraint. We transmit for the first time when $H_k>G_k$ and use wiretap coding in all the $(M+1)$ minislots. We store all the transmitted bits in the key buffer also.
	
	From next slot onwards, we use the first mini-slot for wiretap coding (if $H_k>G_k$) and rest of the mini-slots for transmission via secret key (if $H_k\leq G_k$ use only $M$ minislots for transmission with secret key in slot $k$ and do not use the first minislot, with $\overline{R}_k=\min\left(B_k,MC(P(H_k,G_k))n\right)$. In every slot we remove $R_k$ bits and add $\overline{r}_k\geq R_k$ bits to the key buffer. Since $Pr\left(H_k>G_k\right)>0$,  $Pr(\overline{r}_k>\overline{R}_k)>0$. Thus $B_k \uparrow \infty$ $a.s.$ and eventually, in every slot we will transmit in the first mini-slot at rate 
	\begin{equation}
	\left[C\left(P(H_k,G_k)\right)-C_e\left(P(H_k,G_k)\right)\right]^+
	\end{equation}
	and in the rest of the mini-slots at rate $C\left(P(H_k,G_k)\right)$ with arbitrarily large probability.
	
	The average rate in a slot can be made as close to $C\left(P(H_k,G_k)\right)$ as we wish by making $M$ large enough. Thus, the rate for this coding scheme is maximized by \textit{water filling}.
	
	Now we want to ensure that  for $k$ large enough, for messages $\left(\overline{W}_k,\overline{W}_{k-1},\ldots,\overline{W}_{k-N_1}\right)$ we use only keys from $\left(\overline{W}_1,\ldots,\overline{W}_{k-N_1-1}\right)$. It can be ensured if we do not use more than $\overline{M}$ key bits in a slot and from $k-N_1$ onward the key queue length $\geq \overline{M}N_1$ bits, where the constant $\overline{M}$ can be chosen arbitrarily large. Thus we modify the above scheme such that we use $\min\left(B_k,\overline{M}, nMC(P(H_k))\right)$ key bits in a slot instead of $\min(B_k,nMC(P(H_k)))$ bits. By making $\overline{M}$ as large as needed, we can get arbitrarily close to the water filling rate.
\end{proof}

Strong secrecy can be achieved as for the non-fading case in Section-\ref{section_capacity}. Also, to attain the required reliability and secrecy, the key buffer length required can be obtained as in Section \ref{section_capacity} by using $\overline{M}$ (defined in the proof of Theorem 2).

\section{Fading Wire-tap With no CSI of Eavesdropper}
\label{section_no_csi}
In this section we assume that the transmitter knows only the channel state of Bob at time $k$ but not $G_k$, the channel state of Eve. This is more realistic because Eve is a passive listener. Now we modify our fading model. Instead of $(H_k,G_k)$ being constant during a slot (slow fading), the coherence time of $(H_k,G_k)$ is much smaller than the duration $n$ of a minislot. Then we can use the coding-decoding scheme of \cite{gopala2008secrecy} in the first minislot with secrecy rate $R_s$ and $I(\overline{W}_{k,1};\overline{Z}_1,\ldots,\overline{Z}_k)\leq n\epsilon$, where
\begin{equation}
R_s=\frac{1}{2}\mathsf{E}_{H,G}\left\{\Biggl [ \log\left(1+\frac{HP(H)}{\sigma_1^2}\right)-\log\left(1+\frac{GP(H)}{\sigma_2^2}\right)\Biggr]^+\right\}. \label{secrecy_capacity_fading}
\end{equation}
Now we have the following proposition

\begin{prop}
	Secrecy capacity equal to the main channel capacity  without CSI of Eve at the transmitter
	\begin{equation}
	C=\frac{1}{2}\mathsf{E}_H \left[ \log \left(1+\frac{HP(H)}{\sigma_1^2}\right) \right]
	\end{equation}
	is achievable subject to power constraint $\mathsf{E}_H \left[ P(H) \right] \leq \bar{P}$, where $P(H)$ is the waterfilling policy.
\end{prop}

\begin{proof}
	Since each mini-slot is of long duration compared to the coherence time of the fading process $(H_k,G_k)$, the coding scheme of \cite{gopala2008secrecy} can be used without the CSI of Eve in the first minislot of each slot. This can achieve secrecy capacity 
	\begin{equation}
	C_s=\mathsf{E}_{H,G}\left[\frac{1}{2}\log\left(\frac{1+HP(H)/\sigma_1^2}{1+GP(H)/\sigma_2^2}\right)^{+}\right]
	\end{equation}
	subject to the power constraint $\mathsf{E}_{H,G}\left[P(H)\right] \leq \bar{P}$, with $I(\overline{W}_k;\overline{Z}_k)\leq n\epsilon.$ Now we can use the coding-decoding scheme of Section 3 to achieve the secrecy capacity equal to the main channel capacity
	\begin{equation}
		C=\frac{1}{2}\mathsf{E}_H \left[ \log \left(1+\frac{HP(H)}{\sigma_1^2}\right) \right].
		\end{equation}
\end{proof}

\section{Conclusions}

\label{section_conclution}

In this paper we have achieved secrecy rate equal to the main channel capacity of a wiretap channel by using the previous secret messages as a key for transmitting the current message. We have shown that not only the current message being transmitted, but all messages transmitted in last $N_1$ slots are secure w.r.t. all the outputs of the eavesdropper till now, where $N_1$ can be taken arbitrarily large. We have extended this result to fading wiretap channels when CSI of Eve may or may not be available to the transmitter. The optimal power control is water filling itself.

\section{Appendices}
\section{Proofs of Lemmas}
\textit{Lemma 1}: The following holds
\begin{equation}
I(\overline{W}_{k,1},\overline{W}_{k-1,1},\ldots,\overline{W}_{k-N_1,1};\overline{Z}_1,\overline{Z}_2,\ldots,\overline{Z}_k)\leq (N_1+1)n\epsilon.
\end{equation}
\textit{Proof}: We have,
\begin{align}
&I(\overline{W}_{k,1},\overline{W}_{k-1,1},\ldots,\overline{W}_{k-N_1,1};\overline{Z}_1,\overline{Z}_2,\ldots,\overline{Z}_k)  \nonumber \\
&=I(\overline{W}_{k,1};\overline{Z}_1,\overline{Z}_2,\ldots,\overline{Z}_k)  \nonumber \\
&~~+I(\overline{W}_{k-1,1};\overline{Z}_1,\overline{Z}_2,\ldots,\overline{Z}_k|\overline{W}_{k,1})+\ldots+  \nonumber \\
&~~+I(\overline{W}_{k-N_1,1};\overline{Z}_1,\overline{Z}_2,\ldots,\overline{Z}_k|\overline{W}_{k,1},\ldots, \overline{W}_{k-N_1+1,1}).  \label{eqn_lem1_1}
\end{align}
But
\begin{align}
&I(\overline{W}_{k,1};\overline{Z}_1,\overline{Z}_2,\ldots,\overline{Z}_k) \nonumber \\
&=I(\overline{W}_{k,1};\overline{Z}_{k,1})+I(\overline{W}_{k,1};\overline{Z}_1,\ldots,\overline{Z}_{k-1},\overline{Z}_{k,2}|\overline{Z}_{k,1})  \nonumber \\
&\leq  n\epsilon+0, \label{eqn_lem1_2}
\end{align}
because $(\overline{Z}_{1},\ldots,\overline{Z}_{k-1},\overline{Z}_{k,2}) \perp (\overline{Z}_{k,1}, \overline{W}_{k,1})$, where $X\perp Y$ denotes that random variable $X$ is independent of $Y$.

Next consider
\begin{align}
&I(\overline{W}_{k-1,1};\overline{Z}_1,\overline{Z}_2,\ldots,\overline{Z}_k|\overline{W}_{k,1})  \nonumber \\
&=I(\overline{W}_{k-1,1};\overline{Z}_{k-1,1}|\overline{W}_{k,1}) \nonumber \\
&+I(\overline{W}_{k-1,1};(\overline{Z}_1,\ldots,\overline{Z}_k)-\overline{Z}_{k-1,1}|\overline{W}_{k,1},\overline{Z}_{k-1,1}),   \label{eqn_lem1_3}
\end{align}
where $(\overline{Z}_1,\ldots,\overline{Z}_k)-\overline{Z}_{k-1,1}$ denotes the sequence $(\overline{Z}_1,\ldots,\overline{Z}_k)$ without $\overline{Z}_{k-1,1}$. However, 
\begin{align}
&I(\overline{W}_{k-1,1};\overline{Z}_{k-1,1}|\overline{W}_{k,1}) \nonumber \\
&=I(\overline{W}_{k-1,1};\overline{Z}_{k-1,1}) \leq n\epsilon.  \label{eqn_lem1_4}
\end{align}
Also, because $(\overline{Z}_1,\ldots,\overline{Z}_{k-2})$ is independent of $(\overline{W}_{k-1,1},\overline{W}_{k,1},\overline{Z}_{k-1,1})$,
\begin{align}
&I(\overline{W}_{k-1,1};(\overline{Z}_1,\ldots,\overline{Z}_k)-\overline{Z}_{k-1,1}|\overline{W}_{k,1},\overline{Z}_{k-1,1})  \nonumber \\
&=I(\overline{W}_{k-1,1};\overline{Z}_1,\ldots,\overline{Z}_{k-2}|\overline{W}_{k,1},\overline{Z}_{k-1,1}) \nonumber \\
&~+I(\overline{W}_{k-1,1};\overline{Z}_k,\overline{Z}_{k-1,2}|\overline{W}_{k,1},\overline{Z}_{k-1,1},\overline{Z}_1,\ldots \overline{Z}_{k-2})   \\
&\overset{(a)}{=}0+I(\overline{W}_{k-1,1};\overline{Z}_{k,1}|\overline{W}_{k,1},\overline{Z}_{k-1,1},\overline{Z}_1,\ldots \overline{Z}_{k-2}) \nonumber \\
&~+I(\overline{W}_{k-1,1};\overline{Z}_{k,2},\overline{Z}_{k-1,2}|\overline{W}_{k,1},\overline{Z}_{k-1,1},\overline{Z}_1,\ldots \overline{Z}_{k-2},\overline{Z}_{k,1}).
\end{align}
Furthermore, since $(\overline{W}_{k-1,1},\overline{W}_{k,1},\overline{Z}_{k,1},\overline{Z}_{k-1,1}) \perp (\overline{Z}_1,\ldots, \overline{Z}_{k-2})$
we have
\begin{align}
&I(\overline{W}_{k-1,1};Z_{k,1}|\overline{W}_{k,1},\overline{Z}_{k-1,1},\overline{Z}_1,\ldots, \overline{Z}_{k-2})  \nonumber \\
&=I(\overline{W}_{k-1};Z_{k,1}|\overline{W}_{k,1},\overline{Z}_{k-1,1}).    \label{eqn_lem1_5}
\end{align} 
Using the fact that $(\overline{W}_{k-1},\overline{Z}_{k-1,1}) \perp (\overline{W}_{k,1},\overline{Z}_{k,1})$
we can directly show that the right side equals zero.

Let $A$ denote the indices of the slots in which messages are transmitted which are used as keys for transmitting $\overline{W}_{k,2}$ and $\overline{W}_{k-1,2}$.
Since
\begin{align}
&(\overline{Z}_{k,2},\overline{Z}_{k-1,2})\leftrightarrow (\overline{W}_{k-1,1},\overline{W}_{A}) \nonumber \\
&~~~~~~~~~~~~~~~~~~~~\leftrightarrow (\overline{W}_{k,1},\overline{Z}_{k-1,1},\overline{Z}_{k,1},\overline{Z}_1,\ldots,\overline{Z}_{k-2}), \label{markov1}
\end{align}
where $X\leftrightarrow Y\leftrightarrow Z$ denotes that $\{X,Y,Z\}$ forms a Markov chain, we have
\begin{align}
&I(\overline{W}_{k-1,1};\overline{Z}_{k,2},\overline{Z}_{k-1,2}|\overline{W}_{k,1},\overline{Z}_{k-1,1},\overline{Z}_1,\ldots,\overline{Z}_{k-2},\overline{Z}_{k,1})  \nonumber  \\
&\leq I(\overline{W}_{k-1,1},\overline{W}_{A};\overline{Z}_{k,2},\overline{Z}_{k-1,2}|\overline{W}_{k,1},\overline{Z}_{k-1,1},\notag \\
&\mspace{280mu} \overline{Z}_1,\ldots,\overline{Z}_{k-2},\overline{Z}_{k,1})   \nonumber \\
&\overset{(a)}{\leq} I(\overline{W}_{k-1,1},\overline{W}_A;\overline{Z}_{k,2},\overline{Z}_{k-1,2})   \nonumber \\
\end{align}
\begin{align}
&\leq I(\overline{W}_{k-1,1};\overline{Z}_{k,2},\overline{Z}_{k-1,2})+I(\overline{W}_A;\overline{Z}_{k,2},\overline{Z}_{k-1,2}|\overline{W}_{k-1,1}) \notag \\
&\overset{(b)}{=} 0+I(\overline{W}_A;\overline{Z}_{k,2},\overline{Z}_{k-1,2})\overset{(c)}{=}0,      \label{eqn_lem1_6}
\end{align}
where $(a)$ follows from (\ref{markov1}), $(b)$ follows since $(\overline{W}_{k-1,1},\overline{Z}_{k-1,1})\perp (\overline{W}_A, \overline{Z}_{k,2}, \overline{Z}_{k-1,2})$ and $(c)$ follows since $\overline{W}_A \perp Z_{k,2},Z_{k-1,2}$.

From  (\ref{eqn_lem1_3}), (\ref{eqn_lem1_4}), (\ref{eqn_lem1_5}), (\ref{eqn_lem1_6}),
\begin{equation}
I(\overline{W}_{k-1,1};\overline{Z}_1,\overline{Z}_2,\ldots, \overline{Z}_k|\overline{W}_{k,1})\leq n\epsilon.
\end{equation}

We can similarly show that the other terms on the right side of (\ref{eqn_thm_1}) are also upper bounded by $n\epsilon$. This proves the lemma.$\square$

\textit{Lemma 2}: The following holds
\begin{align}
&I(\overline{W}_{k,2},\overline{W}_{k-1,2},\ldots \overline{W}_{k-N_1,2};\overline{Z}_1,\ldots, \overline{Z}_k| \overline{W}_{k,1},\ldots,\overline{W}_{k-N_1,1})\leq N_1n\epsilon. \label{lemma_2_criteria}
\end{align}
\textit{Proof:} We have
\begin{align}
I(\overline{W}_{k,2}&,\ldots,\overline{W}_{k-N_1,2};\overline{Z}_1,\ldots,\overline{Z}_k|\overline{W}_{k,1},\ldots,\overline{W}_{k-N_1,1})  \nonumber  \\
&= I(\overline{W}_{k,2},\ldots, \overline{W}_{k-N_1,2};\overline{Z}_1,\ldots,\overline{Z}_{k-N_1-1}| \overline{W}_{k,1},\ldots,\overline{W}_{k-N_1,1})   \nonumber \\
&~~+I(\overline{W}_{k,2},,\ldots,\overline{W}_{k-N_1,2};\overline{Z}_{k-N_1},\ldots,\overline{Z}_k| \overline{W}_{k,1},\ldots,\overline{W}_{k-N_1,1},\overline{Z}_1,\ldots,\overline{Z}_{k-N_1-1}).    \label{eqn_lem2_1}
\end{align}

$\quad$Since $\overline{W}_{k,1},\ldots,\overline{W}_{k-N_1,1}$ is independent of $\overline{W}_{k,2},\ldots,\overline{W}_{k-N_1,2},\overline{Z}_1,\ldots,\overline{Z}_{k-N_1-1}$, the first term on the right equals
\begin{equation}
I(\overline{W}_{k,2},\overline{W}_{k-1,2},\ldots,\overline{W}_{k-N_1,2};\overline{Z}_1,\ldots,\overline{Z}_{k-N_1-1})=0. \label{eqn_lem2_2}
\end{equation}
The second term in the RHS of (\ref{eqn_lem2_1})
\begin{align}
&=I(\overline{W}_{k,2},\ldots,\overline{W}_{k-N_1,2};\overline{Z}_{k-N_1},\ldots,\overline{Z}_k| \overline{W}_{k,1},\ldots,\overline{W}_{k-N_1,1},\overline{Z}_1,\ldots,\overline{Z}_{k-N_1-1}) \nonumber \\
&=I(\overline{W}_{k,2},\ldots,\overline{W}_{k-N_1,2};\overline{Z}_{k-N_1,1},\ldots \overline{Z}_{k,1}|  \overline{W}_{k,1},\ldots,\overline{W}_{k-N_1,1},\overline{Z}_1,\ldots,\overline{Z}_{k-N_1-1})\nonumber  \\
&~~+I(\overline{W}_{k,2},\overline{W}_{k-1,2},\ldots,\overline{W}_{k-N_1,2}; \overline{Z}_{k-N_1,2},\ldots,\overline{Z}_{k,2}|\overline{W}_{k,1},\ldots,\overline{W}_{k-N_1,1},\notag \\
&~~~ \overline{Z}_1,\ldots,\overline{Z}_{k-N_1-1},\overline{Z}_{k-N_1,1},\ldots,\overline{Z}_{k,1}).     \label{eqn_lem2_3}
\end{align}

The first term on the right is zero because $(\overline{W}_{k,2}$,$\overline{W}_{k-1,2},\ldots,\overline{W}_{k-N_1,2})$ is independent of $(\overline{Z}_{k,1}$,$\ldots$,$\overline{Z}_{k-N_1,1})$, $(\overline{W}_{k,1},\ldots,\overline{W}_{k-N_1})$ and $\overline{Z}_1$,$\ldots$ $\overline{Z}_{k-N_1-1}.$ Also since $(\overline{W}_{k,1},\ldots,\overline{W}_{k-N_1,1})$ and 
$(\overline{Z}_{k,1},\ldots,\overline{Z}_{k-N_1,1})$ are independent of the other random variables in the second term on the right side, this term equals
\begin{align}
I(\overline{W}_{k,2},\overline{W}_{k-1,2},\ldots,\overline{W}_{k-N_1,2}; \overline{Z}_{k,2},\ldots,\overline{Z}_{k-N_1,2}| \overline{Z_1},\ldots,\overline{Z}_{k-N_1-1}). 
\end{align}
For convenience we denote it as $I(\widehat{W}_2;\widehat{Z}_2|\widehat{Z}_1)$ with $\widehat{W}_2$, $\widehat{Z}_2$, $\widehat{Z}_1$ denoting the respective sequences of random variables. Since
\begin{align}
I(\widehat{W}_2;\widehat{Z}_1,\widehat{Z}_2)&=I(\widehat{W}_2;\widehat{Z}_1)+I(\widehat{W}_2;\widehat{Z}_2|\widehat{Z}_1)  \nonumber \\
&=I(\widehat{W}_2;\widehat{Z}_2)+I(\widehat{W}_2;\widehat{Z}_1|\widehat{Z}_2),
\end{align}
and we have 
\begin{equation}
I(\widehat{W}_2,\widehat{Z}_1)=0=I(\widehat{W}_2;\widehat{Z}_2),   \label{eqn_lem2_0}
\end{equation}
and 
\begin{equation}
\widehat{Z}_1 \leftrightarrow (\widehat{W}_1, \widehat{W}_A,\widehat{W}_2) \leftrightarrow \widehat{Z}_2, \label{markov2}
\end{equation}
where $\widehat{W}_1=(\overline{W}_{k,1},\ldots,\overline{W}_{k-N_1,1})$, we get
\\
\begin{align}
I(\widehat{W}_2;\widehat{Z}_2|\widehat{Z}_1)&\overset{(a)}{=} I(\widehat{W}_2;\widehat{Z}_1|\widehat{Z}_2)  \nonumber \\
&\leq I(\widehat{W}_1,\widehat{W}_2, \widehat{W}_A;\widehat{Z}_1|\widehat{Z}_2) \nonumber \\
&\overset{(b)}{\leq} I(\widehat{W}_1,\widehat{W}_2,\widehat{W}_A;\widehat{Z}_1) \nonumber \\
&=I(\widehat{W}_1;\widehat{Z}_1)+I(\widehat{W}_2, \widehat{W}_A;\widehat{Z}_1|\widehat{W}_1)  \nonumber \\
&\overset{(c)}{=} 0 + I(\widehat{W}_2,\widehat{W}_A;\widehat{Z}_1) \nonumber \\
&=I(\widehat{W}_A;\widehat{Z}_1)+I(\widehat{W}_2;\widehat{Z}_1\lvert \widehat{W}_A) \nonumber \\
&\overset{(d)}{\leq} N_1n\epsilon +0=N_1n\epsilon,
\end{align}
where $(a)$ follow from (\ref{eqn_lem2_0}), $(b)$ follows from (\ref{markov2}) $(c)$ follows from $\widehat{W}_1 \perp (\widehat{W}_2, \widehat{W}_A, \widehat{Z}_1)$, $(d)$ follows from wiretap coding and the fact that the set $A$ will not have larger cardinality than $N_1$,  and $\widehat{W}_2 \perp (\widehat{Z}_1, \widehat{W}_A)$.

Therefore
\begin{equation}
I(\widehat{W}_2;\widehat{Z}_2|\widehat{Z}_1)\leq N_1n\epsilon.  \label{eqn_lem2_4}
\end{equation}

From (\ref{eqn_lem2_2}), (\ref{eqn_lem2_3}) and (\ref{eqn_lem2_4}), we get the lemma.

\section{Acknowledgments}
The first author would like to thank Deekshit PK for useful discussions.

\bibliographystyle{IEEEtran}
\bibliography{single_SHAHID_IET_SINGLE_USER_FEB_11}

\end{document}